\documentclass[10pt]{article}
\usepackage{lmodern} 
\usepackage[utf8,utf8x]{inputenc}
\usepackage[a4paper]{geometry}
\usepackage{amsfonts,amsmath,amssymb,amsthm,latexsym}
\usepackage{indentfirst}
\usepackage{algorithm}

\newcommand{\be}{\begin{enumerate}}
\newcommand{\ee}{\end{enumerate}}
\newcommand{\bi}{\begin{itemize}}
\newcommand{\ei}{\end{itemize}}
\newcommand{\bd}{\begin{description}}
\newcommand{\ed}{\end{description}}
\newcommand{\bc}{\begin{center}}
\newcommand{\ec}{\end{center}}
\newcommand{\cO}{\mathcal{O}}
\newcommand{\N}{\mathbb{N}}
\newcommand{\Z}{\mathbb{Z}}

\def\l{\left}
\def\r{\right}
\newtheorem{defi}{Definition}[section]
\newtheorem{thm}{Theorem}[section]
\newtheorem{lem}{Lemma}[section]
\newtheorem{prop}{Proposition}[section]
\theoremstyle{definition}
\newtheorem{rem}{Remark}
\newtheorem{ex}{Example}
\newtheorem{exs}{Examples}
\def\bibfmta#1#2#3#4{{\sc #1.} {#2}. {\em #3}, #4.}
\def\bibfmtb#1#2#3#4{{\sc #1.} {\em #2}. {#3}, #4.}

\begin{document}

\title{\bf Parallel Improvements on the Accelerated Integer GCD Algorithm}
\author{Sidi Mohamed Sedjelmaci\\
{\small Computer Science Institute, University of Oran Es-Senia. Algeria.}
\and 
Christian Lavault\\
{\small LIPN, CNRS URA 1507, Universit\'e  Paris-Nord, av. J.-B. Cl\'ement F-93430 Villetaneuse}\thanks{Corresponding author. Email: {\tt Christian.Lavault\char64 ura1507.univ-paris13.fr}}
}
\date{\empty}
\maketitle

\begin{abstract}
The present paper analyses and presents several improvements to the algorithm for finding the $(a,b)$-pairs of integers 
used in the $k$-ary reduction of the right-shift $k$-ary integer GCD algorithm. While the worst-case complexity 
of Weber's ``Accelerated integer GCD algorithm'' is $\cO\l(\log_\phi(k)^2\r)$, we show that the worst-case number of iterations of the while loop is exactly $\tfrac 12 \l\lfloor \log_{\phi}(k)\r\rfloor$, where 
$\phi := \tfrac 12 \l(1+\sqrt{5}\r)$.\par
We suggest improvements on the average complexity of the latter algorithm and also present two new faster residual 
algorithms: the sequential and the parallel one. A lower bound on the probability of avoiding the while loop in our 
parallel residual algorithm is also given.

\smallskip \noindent {\em Keywords}: Parallel algorithms; Integer greatest common divisor (GCD); Parallel arithmetic computing; Number theory.
\end{abstract}

\section{Introduction}
Given two integers $a$ and $b$, the {\em greatest common divisor} of $a$ and $b$, or $\gcd(a,b)$, 
is the largest integer which divides both $a$ and $b$. Applications for integer GCD algorithms include computer arithmetic, integer factoring, cryptology and symbolic computation. Since Euclid's algorithm, several GCD algorithms have been proposed. 
Among others, the binary algorithm of Stein and the algorithm of Schonhage must be mentioned. With the recent emphasis on parallel algorithms, a large number of new integer GCD algorithms has been proposed. (See~\cite{Sorenson94} for a brief overview.) Among these is the ``right-shift $k$-ary'' algorithm of Sorenson which generalizes the binary algorithm. 
It is based on the following reduction.

Given two positive integers $u > v$ relatively prime to $k$ (i.e., $u$, $k$ and $v$, $k$ are coprime integers), 
pairs of integers $(a, b)$ can be found that satisfy
\begin{equation} \label{eq:coprime}
au + bv\equiv 0\ (\bmod k)\ \quad \text{with}\ \quad 0 < |a|, |b| < \sqrt{k}.
\end{equation}
If we perform the transformation reduction (also called ``$k$-ary reduction'') 
\[
(u,v) \longmapsto (u',v') = \big(|au + bv|/k, \min(u,v)\big),\]
which replaces $u$ with $u' = |au + bv|/k$, the size of $u$ is reduced by roughly $\tfrac 12  \log_2(k)$ bits.

Note also that the product $uv$ is similarly reduced by a factor of $\Omega(\sqrt{k})$, that is $uv$ is reduced by a factor 
$\ge k/\big(2 \lceil\sqrt{k}\rceil\big) > \frac{1}{2} \sqrt{k} - \frac{1}{2}$ (see~\cite{Sorenson94}). 
Another advantage is that the operation a$u + bv$ allows quite a high degree of parallelization.
The only drawback is that $\gcd(u',v)$ may not be equal to $\gcd(u, v)$ (one can only argue that $\gcd(u,v)$ divides 
$\gcd(u',v)$), whence spurious factors which must be eliminated. 

In the $k$-ary reduction, Sorenson suggests table lookup to find sufficiently small $a$ and $b$ such that 
$au + bv\equiv 0\ (\bmod k)$. 
By contrast, Jebelean~\cite{Jebelean93} and Weber~\cite{Weber95,Weber96} both propose an easy algorithm, which finds small 
$a$ and $b$ satisfying property~\eqref{eq:coprime}. The latter algorithm is called the ``Jebelean-Weber Algorithm'' (or JWA for short) in the paper.

The worst-case time complexity of the JWA is $\cO\l(\log_2(k)^2\r)$. In this paper, we first show in Section~2 that the number of iterations of the while loop in the JWA is exactly $t(k) = \tfrac 12  \lfloor\log_\phi(k)\rfloor$ in the worst  case, which makes Weber's result in~\cite{Weber96} more precise. In Section~3, we suggest improvements on the JWA and 
present two new faster algorithms: the sequential and the parallel residual algorithms. Both run faster than the JWA (at least on the average), and their time complexity is discussed. Finally, Section 4 provides a lower bound on the probability of avoiding the while loop of the parallel residual algorithm.

\section{Worst case analysis of the JWA}
Let us first state the JWA presented in~\cite{Weber96} as the ``Accelerated GCD Algorithm'' or the ``General ReducedRatMod
algorithm''.

\begin{algorithm}
{\bf The Jebelean-Weber Algorithm (JWA)}

\medskip \noindent {\em Input}: $x,\, y > 0$, $k > 1$, $\gcd(k,x) = \gcd(k,y) = 1$

{\em Output}: $(n, d)$ such that $0 < n,\, |d| < \sqrt{k}$ and $ny\equiv dx\ (\bmod k)$

\medskip $c := x/y \bmod k$;

$f_1 = (n_1,d_1) := (k,0)$;

$f_2 = (n_2,d_2) := (c,1)$

\medskip $\quad$\ {\bf while}\ $n_2\ge \sqrt{k}$\ \ {\bf do}

$\quad\ f_1 := f_1 - \lfloor n_1/n_2\rfloor f_2$

$\quad$\ {\bf swap} $(f_1,f_2)$

{\bf return}\ $f_2$
\end{algorithm}
The above version of the JWA finds small $n$ and $d$ such that $ny = dx \bmod k$; $a = -d$ and $b = n$ meet the requirements 
of the $k$-ary reduction, so the need for large auxiliary tables is eliminated. Besides, the output from the JWA satisfies 
property~\eqref{eq:coprime}. (The proof is due to Weber in~\cite{Weber95,Weber96}.)

To prove Theorem~\ref{iter} which gives the number of iterations of the while loop in JWA, we need the following technical Lemma~\ref{fibo} about the Fibonacci numbers.

\begin{lem} \label{fibo}
Let $(F_n)_{n\in \N}$ be the Fibonacci sequence defined by the relation
\[
F_0 = 0,\ F_1 = 1,\ \qquad  F_n = F_{n-1} + F_{n-2}\ \ (n > 2).\]
The following two properties hold:

(i)\ $n= \big\lceil\log_\phi(F_{n+1})\big\rceil$\ \qquad \qquad for $n\ge 0$.

(ii)\ $F_{\lceil n/2\rceil}^2 < F_n < F_{\lceil n/2\rceil+1}^2$\ \quad for $n\ge 3$.
\end{lem}

\begin{proof}
For $n\ge 2$, $\phi^{n-1} < F_{n+1} <  \phi^{n}$ clearly holds, whence property {\em(i)}. Now considering both cases 
when $n$ is either even or odd yields the two inequalities in {\em(ii)}.
\end{proof}

\begin{thm} \label{iter}
The number of iterations of the while loop of the JWA is $t(k) = 	frac 12 \lfloor \log_\phi(k)\rfloor$ in the worst case.
\end{thm}

\begin{proof} By~\cite{Knuth81,Lame1844}, for $0 < u,\, v\le N$, the worst-case number of iterations of the ``Extended Euclidean Algorithm'', or EEA for short, is $\cO\big(\log_\phi(N)\big)$. Moreover, the worst case occurs whenever 
$u = F_{p+1}$ and $v = F_p$. The situation is very similar in the JWA's worst case, which also occurs whenever $u = F_{p+1}$ 
and $v = F_p$. However, the (slight) difference between the EEA and the JWA lies in the exit test. The EEA's exit test 
is $n_{i+1} = 0$, where $n_i = \gcd(u,v)$, whereas the JWA's exit test is 
\begin{equation} \label{ineq}
n_i < \sqrt{k}\le n_{i+1}.
\end{equation}
Let then $j = p - i$, the worst case occurs when $k = F_{p+1}$ and $c = F_{p} = n_0$. The number of iterations
of the while loop is $t = i = p - j$, when $i$ satisfies inequalities~\eqref{ineq}. In that case, $n_i = F_{p-i} = F_j$, 
and the exit test~\eqref{ineq} may be written $F_j^2 < F_{p+1} < F_{j+1}^2$. 
Thus, by Lemma~\ref{fibo}, we have $j = \l\lceil \frac{1}{2} (p+1)\r\rceil$ and $t = p - \l\lceil \frac{1}{2} (p+1)\r\rceil$, which yields
\begin{equation*}
p =\;
\begin{cases}
p/2 - 1 & \text{if $p$ is even, and} \\
	frac 12 (p-1) & \text{if $p$ is odd}
\end{cases}
\end{equation*}
Hence, the worst case happens when $p$ is odd, and $t(k) = \tfrac 12  \lfloor \log_\phi(k)\rfloor$.
\end{proof}
In the JWA, Euclid's algorithm is stopped earlier. Yet, as shown in the above proof, the worst-case inputs 
remain the same when the algorithm runs to completion: $u = F_{p+1}$ and $v = F_p$, i.e. $k = F_{p+1}$ and $c = F_p$.

\begin{ex} Let $(k,c) = (F_{12},F_{11}) = (144,89)$. From the JWA, we get $t = 5$ as expected, 
and $t(144) = \tfrac 12  \lfloor \log_\phi(144)\rfloor = 5$.
\end{ex}
Notice that if $k = 2^{2\ell}$, the worst case never occurs since $k$ cannot be a Fibonacci number. However, the case 
when $k = 2^{2\ell}$ corresponds to how the algorithm is usually used in practice. The actual worst running time 
of the algorithm is still less than its theoretical worst-case number of iterations. 
More precisely, whenever $k = 2^{2\ell} t(k) = \cO\l(\log_2(k)\r)$ only.

\section{Two residual algorithms}
In the sequel, we make use of the following notation: for $k\ge 4$, $A_k$, $B_k$ and $U_k$ are the sets of positive
integers defined by
\[
A_k = ]0,\sqrt{k}[,\ \quad B_k = ]k-\sqrt{k},k[,\ \quad U_k = A_k\cup B_k.\]

\begin{defi}
Let $(x,y)\in U_k\times U_k$. The $T$-transformation is defined as follows.
\bd
\item If $x,y\in A_k$, then $T(x,y) = (x,y)$.
\item If $x\in A_k$ and $y\in B_k$, then $T(x,y) = (x,y-k)$.
\item If $x\in B_k$ and $y\in A_k$, then $T(x,y) = (k-x,-y)$.
\item If $x,y\in B_k$, then $T(x,y) = (k-x,k-y)$.
\ed
\end{defi}

\begin{rem}
The (equivalent) analytic definition of the $T$-transformation is
\[
T(x,y) = \big( (1-2\chi(x))x + \chi(k)k, (1-2\chi(x))(y-\chi(y)k) \big),\]
where $\chi$ is the characteritic function of the set $B_k$.
\end{rem}

\begin{prop}
For every $(x,y)\in U_k\times U_k$, the pair $(x',y') = T(x,y)$ satisfies

$(i)$\ $0 < x,\, |y'| < \sqrt{k}$.

$(ii)\ x'y\equiv xy'\ (\bmod k)$.
\end{prop}

\begin{proof} \hfill

$(i)$\ If $k - \sqrt{k} < x < k$, then $0 < k - x < \sqrt{k}$ and\ $0 < |k - x| < \sqrt{k}$.

$(ii)$ is easily derived from the definition of $T$.
\end{proof}

\subsection{The residual algorithm}

\begin{algorithm}
{\bf The residual Algorithm {\em Res}}

\medskip \noindent {\em Input}: $x,\, y > 0$, $k > 1$, $\gcd(k,x) = \gcd(k,y) = 1$

{\em Output}: $(n, d)$ such that $0 < n,\, |d| < \sqrt{k}$ and $ny\equiv dx\ (\bmod k)$

\medskip $a := x \bmod k$; $b := y \bmod k$

\medskip $\quad$\ {\bf if}\ $(a,b)\in U_k\times U_k$\ \ {\bf then}\ \ $f_2 := T(a,b)$\ \ {\bf else}

$\qquad$\ $f_1 = (n_1,d_1) := (k,0)$

$\qquad$\ $f_2 = (n_2,d_2) := (c,1)$

\medskip $\quad$\ {\bf while}\ $n_2\ge \sqrt{k}$\ {\bf do}

$\qquad\ f_1 := f_1 - \lfloor n_1/n_2\rfloor f_2$

$\qquad$\ {\bf swap} $(f_1,f_2)$

{\bf return}\ $f_2$ \kern2cm /*~$f_2 = Res(x,y)$~*/
\end{algorithm}

The worst-case complexity of the residual algorithm remains in the same order of magnitude as the JWA, 
$\cO\l(\log_2(k)^2\r)$. However, the above algorithm runs faster on the average. The use of transformation $T$ 
makes it possible to avoid the while loop quite often indeed. (See the related probability analysis in Section~4.)
For example, the residual algorithm provides an immediate result in the cases when $(a,b)\in U_k\times U_k$ or 
$c > k - \sqrt{k}$.

Note that the computational instruction $c := a/b \bmod k$ may be performed either by the euclidean algorithm~\cite{GrKr84}, or by a routine proposed by Weber when $k$ is an even power of two~\cite{Weber95,Weber96}.
Since $x$ and $y$ are symmetrical variables, the same algorithm can also be designed with the instruction $s := b/a \bmod k$ instead of $c := a/b \bmod k$, and then by swapping $n$ and $d$ at the end of the algorithm. This remark leads to an obvious improvement about the residual algorithm: why not compute in parallel both $c$ and $s$? The following parallel algorithm 
is based on such an idea.

\newpage
\subsection{The parallel residual algorithm}

\begin{algorithm}
{\bf The Parallel Residual Algorithm {\em Pares}}

\medskip \noindent {\em Input}: $x,\, y > 0$, $k > 1$, $\gcd(k,x) = \gcd(k,y) = 1$

{\em Output}: $(n, d)$ such that $0 < n,\, |d| < \sqrt{k}$ and $ny\equiv dx\ (\bmod k)$

\medskip $a := x \bmod k$; $b := y \bmod k$

\medskip $\quad$\ {\bf if}\ $(a,b)\in U_k\times U_k$\ \ {\bf then}\ \ $f_2 := T(a,b)$\ \ {\bf else pardo}

$\qquad$\ $v_1 = Res(a,b)$; $v_2 = Res(b,a)$

{\bf return}\ $f_2$
\end{algorithm}
$v_1$ and $v_2$ are two variables whose values are the result returned in the parallel computation performed by $Res(a,b)$ and $Res(b,a)$, respectively. The algorithm {\em Pares} ends when either of these two algorithms terminates.
 
$Res(a,b)$ is the residual algorithm described in \S3.1 and $Res(b,a)$ is the following (very slightly) modified version 
of {\em Res}.

\begin{algorithm}

$a := b/a \bmod k$;

\medskip {\bf if}\ $s\in U_k$\ \ {\bf then}\ \ $f_2 := T(1,s)$\ \ {\bf else}

$\quad\ f_1 = (n_1,d_1) := (k,0)$

$\quad\ f_2 = (n_2,d_2) := (s,1)$

\medskip $\qquad$\ {\bf while}\ $n_2\ge \sqrt{k}$\ {\bf do}

$\qquad\ \ \quad\ f_1 := f_1 - \lfloor n_1/n_2\rfloor f_2$

$\qquad\ \ \quad$\ {\bf swap}\ $(f_1,f_2)$

$\qquad$\ {\bf endwhile}

\medskip $\quad$\ {\bf swap}\ $(n_2,d_2)$

{\bf return}\ $f_2$
\end{algorithm}

\begin{rem}
$Res(b,a)$ and $Res(b,a)$ are the only parallel routines performed in the algorithm {\em Pares}, and they are both 
to terminate if either one or the other finishes. Such a (quasi-) serial computation certainly induces an overhead on most parallel computers. Overhead costs may yet be reduced to a minimum thanks to a careful scheduling and synchronization 
of tasks and processors.

Note also that $s$ may belong to $U_k$ while $c$ does not. This may be seen in the following example.
\end{rem}

\begin{ex}
Let $k = 1024$, $(a,b) = (263,151)$, and $\sqrt{k} = 32$. Then, $c = a/b \bmod k = 273$, and $c\notin U_k$. 
But $s = b/a \bmod k = 1009\in U_k$. So, the while loop is simply avoided by performing $f_2 := T(1,1009)$. 
Such an example shows that the parallel residual algorithm is very likely to run faster than its sequential variant, 
at least on the average.
\end{ex}

\section{Probability analysis}
We first need a technical result to perform the evaluation of the probability that the while loop is avoided in the parallel residual algorithm. 

\begin{lem} \label{ineqs}
Let $k$ be a square such that $k\ge 9$, and let
\[
E_k = \big\{x\in \N\; |\; 1\le x\le k\ and\ \gcd(x,k) = 1\big\}.\]
Then, for every $x\in E_k$\ and $1 < x < \sqrt{k}$, we have
\begin{equation} \label{eq:ineqs}
\sqrt{k} < 1/x \bmod k < k - \sqrt{k}.
\end{equation}
\end{lem}

\begin{proof}
Notice first that, obviously, there cannot exist any integer $1 < x < 4$ for $k = 1$ and $k = 4$; whence
the statement of the lemma: $k\ge 9$.

Let $x\in E_k$ such that $1 < x < \sqrt{k}$\ and set $y = 1/x \bmod k\in E_k$. The whole proof is by contradiction.

First, on the assumption that $1 < x < \sqrt{k}$, suppose that $y\le \sqrt{k}$. Hence, $xy < k$ and since 
$xy\equiv 1\ (\bmod k)$ with $x > 1$, the contradiction is obvious. Thus $y = 1/x \bmod k > \sqrt{k}$.

Now, let us prove that $y < k - \sqrt{k}$ in Eq.~\eqref{eq:ineqs}. On the assumption that $1 < x < \sqrt{k}$, 
suppose also by contradiction that $y\ge k - \sqrt{k}$, with $\gcd(y,k) = 1$ and $y\le k$. Let $m$, $n$ be two non-negative integers, and let $x = \sqrt{k} - m$, where $1\le m\le \sqrt{k} - 2$ and $y = k - \sqrt{k} + n$, where $0\le n\le \sqrt{k}$.
The upper bound on $n$ may be reduced as follows: $n\neq \sqrt{k}$, since if $y = k$, $\gcd(y,k)\neq 1$ and $y\notin E_k$.
So that $0\le n\le \sqrt{k} - 1$.

The product $xy$ writes
\[
xy = (\sqrt{k} - m)(k - \sqrt{k} + n) = k(\sqrt{k} - m) + P(m,n) + 1 - k,\]
where $P(m,n) = k - 1 - (\sqrt{k} - m)(\sqrt{k} + n)$.

Now we have that $xy\equiv 1\ (\bmod k$) and, therefore, $P(m,n)$ must satisfy
\begin{equation} \label{eq:pmn}
P(m,n)\equiv 0\ (\bmod k).
\end{equation}
From the bounds on $m$ and $n$ we can derive bounds on $P(m,n)$, 
\begin{flalign*}
k - 1 - (\sqrt{k} - 1)\sqrt{k} & \le P(m,n)\le k - 1 - \big(\sqrt{k} - (\sqrt{k}-2)\big) \big(\sqrt{k} - (\sqrt{k}-1)\big)\\
\sqrt{k} - 1 & \le P(m,n)\le k - 3 
\end{flalign*}
and, since $k\ge 9$, $1 < P(m,n) < k$: a contradiction with Eq.~\eqref{eq:pmn}.
\end{proof}

\begin{rem}
Lemma~\ref{ineqs} is false when $k$ is not a square: e.g., let $k = 17$. Then $x = 4$ and $y = 1/x \bmod k = 13$, 
while $k - \sqrt{k} < 17 - 4 = 13$.
\end{rem}

\begin{prop} \label{oneone}
Let k be a square such that $k\ge 9$. Let $\lambda$ be a one-one mapping, $\lambda : E_k\longleftrightarrow E_k$, defined by
$\lambda(x) = 1/x \bmod k$. Then we have
\be
\item[{\em (i)}]\ $U_k \cap \lambda(U_k) = \{1,k-1)\}$.
\item[{\em (ii)}]\ $|U_k \cup \lambda(U_k)| = 4\varphi(\sqrt{k}) - 2$,
\ee
where $\varphi$ denotes Euler's totient function $\varphi(m) = \l|(\Z/k\Z)^{*}\r|$ defined for any integer $m\ge 1$.
\end{prop}

\begin{proof} Recall that 
\begin{flalign*}
E_k & = \{x\in \N \,|\, 1\le x\le k\ \text{and}\ \gcd(x, k) = 1\},\\
U_k & = \{x\in E_k \,|\, O < x < \sqrt{k}\ \text{or}\ k\sqrt{k} < x < k\},
\end{flalign*}
and
\[
\lambda(U_k) = \{y\in E_k \,|\, y = 1/x \bmod k\}\subset E_k.\]

{\em (i)}\ Obviously, $1$ and $k - 1$ belong to $U_k$. Let $x\in U_k$, such that $x\neq 1$ and $x\neq k - 1$. 
By definition, $x$ may belong to either distinct subset of $U_k$:

{\em Case~1:}\ $1 < x < \sqrt{k}$. 
By Lemma~\ref{ineqs}, $\lambda(x)\notin U_k$\ and $\lambda(\lambda(x))\notin U_k$.

{\em Case~2:}\ $k - \sqrt{k} < x < k - 1$. 
Let $x' = k - x$, the integers $x$ and $x'$ play a symmetrical role, which brings back to Case~1, 
and $\lambda(x')\notin U_k$.\par
Hence, $\lambda(x') = \lambda(k-x) = k - \lambda(x)\notin U_k$. It follows that $\lambda(x)\notin U_k$ 
and $x = \lambda(\lambda(x))\notin U_k$. Therefore, every integer $x\in U_k$ distinct from 1 and $k - 1$ 
does not belong to $\lambda(U_k)$, and equality {\em (i)}\ follows.

\medskip {\em (ii)}\ The function $\lambda$ being one-to-one, $|\lambda(U_k)| = |U_k|$, which yields
\begin{flalign*}
|U_k \cup \lambda(U_k)| & = |U_k| + |\lambda(U_k)| - |U_k \cap \lambda(U_k)|\\
& = |U_k| - |U_k \cap \lambda(U_k)|.
\end{flalign*}
Now $x < \sqrt{k}$ and $\gcd(x,k) = 1$, so $\gcd(x,\sqrt{k}) = 1$, and thus,
\[
|\{x\in E_k \,|\, \gcd(x,k) = 1, x < \sqrt{k}\}| = \varphi(\sqrt{k})\ \text{and}\ |U_k| = 2\varphi(\sqrt{k}).\]
By equality {\em (i)}, $|U_k \cap \lambda(U_k)| = |\{1,k-1)\}| = 2$ and {\em ((ii)}\ holds: 
$|U_k \cup \lambda(U_k)| = 4\varphi(\sqrt{k}) - 2$.
\end{proof}
From the previous results we can estimate the probability $p_1$ that $x\in U_k$ or $1/x \bmod k\in U_k$ when $k$ is a square $(k\ge 9)$. In particular we have the following theorem.

\begin{thm} \label{p1}
Let $k$ be a square such that $k\ge 9$, and $p_1 = \mathbb{P}\big(x\in U_k\ \text{or}\ 1/x \bmod k\in U_k\big)$. Then,
\[
p_1 = \tfrac{2}{\sqrt{k}}\, \l(2 - \tfrac{1}{\sqrt{k}}\r)\]
\end{thm}

\begin{proof}
$E_k = \{x\in \N \,|\, 1\le x\le k\ \text{and}\ \gcd(x, k) = 1\}\ \ \text{and}\ \ |E_k| = \varphi(k)$. 

Let $x\in E_k$. If $x\notin U_k$ and $\lambda(x) = 1/x \bmod k\in U_k$, $x = \lambda(\lambda(x))\in A(U_k)$.

Now, $|\lambda(U_k)| = |U_k| = 2\sqrt{k}$. Let $r$ be the number of integers $x\in E_k$ such that $x\in U_k$ 
or $1/x \bmod k\in U_k$. By Proposition~\ref{oneone}, $r = |U_k \cup \lambda(U_k)| = 4\varphi(\sqrt{k}) - 2$, 
and $p_1 = r/\varphi(k)$. Since $\varphi(k) = \sqrt{k} \varphi(\sqrt{k}$, the result follows.
\end{proof}

\begin{rem}
Among all possible values of $k$, the case if $k = 2^{2\ell}$ is especially interesting since it allows easy hardware routines. If $k = 2^{2\ell}$ , $\ell\ge 2$, $k$ is a square $\ge 9$ and Thm.~\ref{p1} applies. 
Since $\varphi\l(\sqrt{2^{2\ell}}\r) = \varphi\l(2^{\ell}\r) =  2^{\ell-1}$, and
\[
p_1 = 1/2^{\ell-2} - 1/2^{2\ell-2}.\]
\end{rem}

\begin{exs} \hfill

\be
\item Let $k = 16$. We have

\medskip \begin{table}[h!b]
\bc
\begin{tabular}{l|c c c c c c c c}
\hline
$x$ & 1 & 3 & 5 & 7 & 9 & 11 & 13 & 15\\
\hline
$1/x \bmod k$ & 1 & 11 & 13 & 7 & 9 & 3 & 5 & 15\\
\hline
\end{tabular}
\caption{Values of $1/x \bmod 16$ for the 8 first odd integers.}
\label{tab}
\ec
\end{table}
In Table~\ref{tab}, 1, 3, 13, 15 $\in U_{16}$, and also $1/5 \bmod 16$, $1/11 \bmod 16\in \lambda(U_{16})$. 
Thus, the while loop is avoided 6 times (at least) among the 8 possible cases, and $p_1 = 6/8$.\par
Similarly, by Thm.~\ref{p1}, $r = 4\varphi(4) - 2 = 6$, $p_1 = 	\tfrac 12 (2 - \tfrac 12) = 3/4$. In that case, the while loop is avoided 75~\% of the time.

\item Let $k = 64$.  $U_{64}\cup \lambda(U_{64}) = \{$1, 3, 5, 7, 9, 13, 21, 43, 51, 55, 57, 59, 61, 63$\}$: $r = 14$\ 
and $p_1 = 14/32 = (2^3 - 1)/2^4 = 7/16$. The while loop is avoided 14 times among the 32 possible cases, which corresponds to 43.75~\% at least.\par
It is worthwhile to notice that if $c = 39$, then $c\notin U_{64}$\ and $s = 1/c \bmod 64 = 23\notin U_{64}$. In some particular cases however, the while loop can still be avoided. This happens for example when $(a,b) = (3,5)$: 
$c = 39\notin U_{64}$, $s = 23\notin U_{64}$; yet the while loop is avoided since $(3,5)\in U_{64}\times U_{64}$.

\item Let $k = 2^{16}$\ or $k = 2^{32}$. When dealing with 16-bits words, $p_1 = (2^8 - 1)/2^{14}\cong 1.55~\%$. 
With 32-bits words, $p_1 = (2^{16} - 1)/2^{30} = 6\times 10^{-3}~\%$.
\ee
\end{exs}
This latter examples show that $p_1$ is only a lower bound on the actual probability $p$ of ``systematically'' avoiding the while loop, at each iteration of the parallel residual algorithm.

\section{Summary and remark}
We proved that the number of iterations of the while loop in the worst case of the Jebelean-Weber algorithm equals 
$t(k) = \tfrac 12 \lfloor \log_\phi(k)\rfloor$. We presented two new algorithms, the sequential and the parallel residual
algorithm, which both run faster than the JWA (at least on the average). Preliminary experimentations on these algorithms meet the above results and confirm the actual and potential efficiency of the method. A lower bound on the probability of avoiding the while loop of the parallel residual algorithm was also given. 

These improvements have certainly more effect when $k$\ is small, and this is precisely the case when using table-lookup is more efficient than the use of the JWA. However, even if such improvements might seem negligible for only a few iterations of our algorithm, avoiding the inner loop several times repeatedly makes them significant indeed in the end.

\bibliographystyle{article}
\def\bibfmta#1#2#3#4{ {#1}, {#2}, \emph{#3}, #4.}
\bibliographystyle{book}

\begin{thebibliography}{99}
\bibliographystyle{plain}
% 
\bibitem{GrKr84}\bibfmtb
{R.T. Gregory and E.V. Krishnamurthy}
{Methods and Application of Error-Free Computation}{Springer}{1984}
%
\bibitem{Jebelean93}\bibfmta 
{T. Jebelean}
{A generalization of the binary GCD algorithm}
{in Proc. Int. Sympp. on Symbolic and Algebraic Computation (ISSAC'93)}{(1993), pp.~111--116}
%
\bibitem{Knuth81}\bibfmtb
{D.E. Knuth}
{The Art of Computer Programming: Seminumerical Algorithms}{vol. 2, 2nd ed.}{Addison Wesley, 1981}
%
\bibitem{Lame1844}\bibfmta
{G. Lamé}
{Note sur la limite des diviseurs dans la recherche du plus grand commun diviseur entre deux nombres entiers}
{C.R. Acad. Sci. Paris}{19 (1844), pp.~867--870}
%
\bibitem{Sorenson94}\bibfmta
{J. Sorenson}
{Two fast GCD algorithms}{J. Algorithms}{16 (1994), pp.~110--144}
%
\bibitem{Weber95}\bibfmta
{K. Weber}
{The accelerated integer GCD algorithm}{Dept. of Mathematics and Computer Science}{((1995), Kent State Un.}
%
\bibitem{Weber96}\bibfmta
{K. Weber}
{Parallel implementation of the accelerated integer GCD algorithm}{J. Symbolic Comput. 
(Special Issue Parallel Symbolic Computation)}{(1996), to appear}

\end{thebibliography}
\def\bibfmtb#1#2#3#4{ {#1}, \emph{#2}, {#3}, #4.}

\end{document}